\documentclass[twocolumn,superscriptaddress,prl,floatfix]{revtex4-2}

\usepackage{graphicx}
\usepackage{amsmath}
\usepackage{amsthm}
\usepackage{amssymb}
\usepackage{latexsym}
\usepackage{array}
\usepackage{hyperref}
\usepackage{amsfonts}
\usepackage{dsfont}
\usepackage{mathrsfs}
\usepackage{verbatim}
\usepackage{bbold}
\usepackage[normalem]{ulem}
\usepackage{upgreek}
\usepackage{makecell}
\usepackage{adjustbox}
\usepackage{algorithm}
\usepackage{algpseudocode}
\usepackage{color}
\usepackage{bm}
\usepackage{times}
\usepackage{booktabs}
\usepackage{multirow}

\newcommand{\abs}[1]{\left|#1\right|}
\newcommand{\norm}[1]{\left\|#1\right\|}
\newcommand{\ket}[1]{\left|#1\right\rangle}
\newcommand{\bra}[1]{\left\langle #1\right|}

\newcommand{\ketbra}[2]{\ket{#1}\!\!\bra{#2}}
\newcommand{\Tr}{\operatorname{Tr}}
\newcommand{\hI}{\hat \openone}
\newcommand{\hrho}{\hat\rho}

\newcommand{\hU}{\hat U}
\newcommand{\hH}{\hat H}

\newcommand{\hD}{\hat{\mathcal D}}
\newcommand{\hcalU}{\hat{\mathcal U}}
\newcommand{\hM}{\hat{\mathcal M}}
\newcommand{\hE}{\hat{\mathcal E}}

\newcommand{\C}{\mathcal C}
\newcommand{\JS}{\mathrm{JS}}
\newcommand{\HV}{\mathrm{HV}}
\newcommand{\NSIT}{\mathrm{NSIT}}

\newcommand{\KL}{\mathrm{KL}}
\newcommand{\SM}{Supplemental Material}

\theoremstyle{plain}
\newtheorem{theorem}{Theorem}
\newtheorem{proposition}{Proposition}

\newtheorem{definition}{Definition}

\begin{document}

\title{Exact No-Signaling-in-Time without Temporal Classicality}

\author{Minsu~Kim}
\thanks{These authors contributed equally to this work.}
\affiliation{Quantum Sensors Research Center, Pusan National University, Busan 46241, Republic of Korea}

\author{Jeongho~Bang}\email{jbang@yonsei.ac.kr}
\thanks{These authors contributed equally to this work.}
\affiliation{Institute for Convergence Research and Education in Advanced Technology, Yonsei University, Seoul 03722, Republic of Korea}
\affiliation{Department of Quantum Information, Yonsei University, Incheon 21983, Republic of Korea}

\author{Han~Seb~Moon}\email{hsmoon@pusan.ac.kr}
\affiliation{Quantum Sensors Research Center, Pusan National University, Busan 46241, Republic of Korea}
\affiliation{Department of Physics, Pusan National University, Busan 46241, Republic of Korea}
\affiliation{Quantum Science Technology Center, Pusan National University, Busan 46241, Republic of Korea}

\date{\today}

\begin{abstract}
No-signaling-in-time (NSIT) has often been treated as the clean operational remnant of noninvasive measurability: if an earlier measurement leaves every later marginal unchanged, the temporal process appears classical. We show that this inference is false. We introduce common fixed points (CFPs) of nonselective measurement channels as an exact mechanism that erases all marginal evidence of invasiveness while preserving disturbed outcome-conditioned branches. For a qutrit ring subject to a local degenerate L\"uders measurement, we solve the full CFP manifold analytically. Every state on this manifold satisfies exact pairwise NSIT, yet every nontrivial member violates a Leggett--Garg inequality, including the maximally mixed state. The violation is governed by finite branch displacement, not by residual signaling. Hidden-variable reconstruction, entropic witnesses, protocol-landscape scans, noise robustness, and finite-shot simulations show that exact NSIT certifies only the disappearance of marginal signals after outcome erasure, not the existence of a classical temporal history.
\end{abstract}

\maketitle


{\em Introduction.}---A temporal test of quantumness begins with a conceptual trap. In a spatial Bell inequality scenario, one can ask whether separated outcomes admit a classical explanation. In a sequential experiment scenario, the act of asking the first question may itself prepare the answer to the next one. The macrorealism programs formalized this problem through inequalities, called Leggett--Garg inequalities (LGIs), satisfied by state-evolution histories with noninvasive measurability~\cite{LeggettGarg1985,Emary2014}. No-signaling-in-time (NSIT) later appeared to offer a more direct operational cure: compare the later distribution with and without an earlier measurement, after the earlier outcome has been ignored~\cite{KoflerBrukner2013,ClementeKofler2015}. If no temporal signal remains, what could still be invasive? This question compresses two logically different ideas: the operational disappearance of a signal in an averaged distribution and the macrorealist demand that the earlier interrogation reveal a pre-existing value without altering the subsequent history.

The literature has clarified parts of this logic, but it has left an interpretational ambiguity. NSIT was proposed as a necessary condition beyond LGIs~\cite{KoflerBrukner2013}; suitable complete NSIT packages, supplemented by the causal consistency condition that later choices cannot affect earlier marginals, can be necessary and sufficient for macrorealism~\cite{ClementeKofler2015}; and no set of LGIs alone can play the role that Bell inequalities play in Fine's theorem~\cite{ClementeKofler2016,Fine1982}. Experiments and theory have also shown that the temporal correlations are highly sensitive to measurement design, system dimension, and signaling constraints~\cite{BudroniEmary2014,Knee2012,Robens2015,Knee2016,Formaggio2016,Kreuzgruber2024,Emary2017,Wang2018}. Yet, these advances do not answer the microscopic question that makes NSIT tempting in the first place: when NSIT is exactly zero, has the measurement invasiveness really disappeared, or has it only been hidden by outcome erasure?~\footnote{This actually matters because NSIT is tested only after selective branches have been recombined.}.

We resolve this ambiguity by identifying a channel-level mechanism in which the marginal silence and branch disturbance coexist exactly. A measurement writes a classical outcome register, and ignoring this register produces a nonselective channel. The loophole is not an experimental imperfection but the exact erasure of the marginal evidence of invasiveness by the channel itself. To explain this, we introduce the notion of common fixed points (CFPs). A CFP is an ensemble that is restored by the combined evolution and outcome-erased measurement, even though the selective branches composing that ensemble may be nontrivially different. Specifically, we show that CFPs of a channel can make every pairwise NSIT test exactly silent, while the erased register has nevertheless prepared different selective branches. LGIs, hidden-variable completions, and entropic witnesses are sensitive to the future correlations of these branches. We thus claim that the correct implication chain is
\begin{eqnarray}
NSIT \not\Rightarrow \textrm{\em Noninvasive} \not\Rightarrow \textrm{\em Classical History}.
\label{eq:chain}
\end{eqnarray}
The first arrow is the central message: a silent marginal need not mean that every selective branch is left untouched; it may only mean that the branch disturbances cancel after the outcome register is erased. Our contribution is therefore not a small addition to the NSIT--LGI hierarchy; it is an exact dynamical mechanism showing where the usual inference from marginal silence to temporal classicality breaks down.



{\em Operational framework and CFP theorem.}---Let $\hrho$ be a density matrix and
\begin{eqnarray}
\hat{\mathcal U}_{\tau}(\hrho) = \hU(\tau) \hrho \hU^{\dagger}(\tau),
\quad
\hU(\tau) = \exp(-i\hH \tau / \hbar). 
\label{eq:unitary}
\end{eqnarray}
A dichotomic L\"uders instrument has the projectors $\hat{P}_{\pm}$ with $\hat{P}_{+} + \hat{P}_{-} = \hI$. The selective branch and nonselective channel are
\begin{eqnarray}
\hrho_{q | \hrho} = \frac{\hat P_q\hrho\hat P_q}{p_q},
\quad
\hD(\hrho) = \sum_{q=\pm} \hat{P}_q \hrho \hat{P}_q.
\label{eq:luders}
\end{eqnarray}
where $p_q = \Tr(\hat{P}_q \hrho)$. For times $t_1<t_2<t_3$, a pairwise context $(i,j)$ gives the joint probability $P^{(ij)}(q_i,q_j)$. The later marginal is $P^{(ij)}(q_j)=\sum_{q_i}P^{(ij)}(q_i, q_j)$, while $P^{(j)}(q_j)$ denotes the experiment in which only the later measurement is performed. Here, we define~\cite{KoflerBrukner2013}
\begin{eqnarray}
&& \NSIT_{i \to j} : P^{(ij)}(q_j) = P^{(j)}(q_j), \quad (q_j=\pm), \nonumber \\
&& \eta_{\NSIT} = \max_{i < j,q}\abs{P^{(ij)}(q) - P^{(j)}(q)}.
\label{eq:nsit}
\end{eqnarray}
NSIT is thus a marginal statement after outcome erasure. Note that it does not say that each branch in Eq.~(\ref{eq:luders}) is unperturbed.

LGI probes a different object. With $\hat{Q} = \hat{P}_{+} - \hat{P}_{-}$, let us consider a temporal correlation function,
\begin{eqnarray}
C_{ij} = \sum_{q_i,q_j = \pm1} q_i q_j P^{(ij)}(q_i, q_j).
\label{eq:cij}
\end{eqnarray}
If a single context-independent distribution $P(q_1, q_2, q_3)$ exists, then~\cite{LeggettGarg1985,Fine1982}
\begin{eqnarray}
K_1 = C_{12} + C_{23} - C_{13} \le 1
\label{eq:lgi}
\end{eqnarray}
and the other sign variants also hold. Thus, LGI violation excludes a single classical temporal history reproducing the observed pairwise contexts~\cite{Fine1982,ClementeKofler2016}.

\begin{definition}[Common Fixed Point]
Consider a nonselective sequential channel
\begin{eqnarray}
\hM_{\tau} = \hD \circ \hcalU_{\tau}.
\label{eq:Mdef}
\end{eqnarray}
For a set of protocol intervals $\mathcal{T}$, a state $\hrho^*$ is a common fixed point (CFP) if $\hM_{\tau}(\hrho^*)=\hrho^*$ for all $\tau \in \mathcal{T}$; it is global if this holds for every $\tau\ge0$.
\end{definition}

Physically, CFPs should not be confused with the measurements that leave every possible state unchanged. Such a demand would make an informative readout essentially trivial by the information--disturbance tradeoff. In a sequential test, the relevant object is instead the combined measurement--evolution channel after the outcome register is erased. A CFP is an ensemble invariant under that channel: its averaged state is restored, while the selective outcomes may still carry information and prepare distinct branches. Thus, CFPs realize effective noninvasive measurability only at the ensemble/channel level, not at a single macrorealist history.

We then provide the following theorem:

\begin{theorem}[CFP $\Rightarrow$ exact pairwise NSIT]
If $\hat\rho^*$ is a CFP for every interval connecting the measured times and for the zero-time readout dephasing, $\hM_0(\hat{\rho}^*)=\hD(\hat{\rho}^*)=\hat{\rho}^*$, then Eq.~(\ref{eq:nsit}) holds in every pairwise context and $\eta_{\NSIT}=0$.
\end{theorem}

\begin{proof}---For an arbitrary density matrix $\hat\sigma$, the interval channel is
\begin{eqnarray}
\hM_{\tau}(\hat\sigma)=\sum_{r=\pm}\hat P_r\hU(\tau)\hat\sigma\hU^{\dagger}(\tau)\hat P_r.
\end{eqnarray}
Thus, the CFP condition gives the later readout marginal with the unitary step made explicit,
\begin{eqnarray}
\Tr[\hat{P}_q \hat{\rho}^*] &=& \Tr[\hat{P}_q \hM_{\tau}(\hat{\rho}^*)] \nonumber\\
	&=& \Tr[\hat{P}_q \hU(\tau)\hat{\rho}^*\hU^{\dagger}(\tau)] \equiv P_{\tau}(q). 
\end{eqnarray}
In a pairwise NSIT test, the zero-time readout condition gives $\hD(\hat\rho^*)=\hat\rho^*$, so summing over the earlier outcome yields
\begin{eqnarray}
P^{(ij)}(q_j) &=& \Tr[\hat{P}_{q_j} \hU(\tau_{ij}) \hD(\hat{\rho}^*) \hU^{\dagger}(\tau_{ij})]  \nonumber \\
	&=& \Tr[\hat{P}_{q_j} \hU(\tau_{ij}) \hat{\rho}^* \hU^{\dagger}(\tau_{ij})] = P^{(j)}(q_j).
\end{eqnarray}
Therefore, Eq.~(\ref{eq:nsit}) holds for every pairwise context and $\eta_{\NSIT}=0$. The condition is sufficient, not necessary, because NSIT constrains only observed marginals and not the selective branches.
\end{proof}
The theorem proves exact marginal non-signaling, not branch non-invasiveness. We will quantify the hidden branch disturbance by
\begin{eqnarray}
B(\hrho) = \sum_{q=\pm} p_q \frac{1}{2} \norm{\hrho_{q | \hrho} - \hrho}_1.
\label{eq:Bdef}
\end{eqnarray}
The central regime of this Letter is $\eta_{\NSIT}=0$ with $B(\hrho)>0$.


{\em Qutrit ring and full CFP manifold.}---Consider a single excitation hopping coherently on a three-site ring~\cite{Yi2024},
\begin{eqnarray}
\hH = \hbar\Omega \bigl( \ketbra{1}{0} + \ketbra{2}{1} + \ketbra{0}{2} + \textrm{H.c.} \bigr),
\label{eq:H}
\end{eqnarray}
with the local degenerate measurement
\begin{eqnarray}
\hat{P}_{+} = \ketbra{0}{0},
\quad
\hat{P}_{-} = \hI - \hat{P}_{+},
\quad
\hat{Q} = \hat{P}_{+} - \hat{P}_{-}.
\label{eq:P}
\end{eqnarray}
This measurement asks only whether the excitation occupies site $0$; it does not distinguish sites $1$ and $2$. Here, we introduce
\begin{eqnarray}
\ket{a} =\ket{0},
\quad
\ket{b}=\frac{\ket{1}+\ket{2}}{\sqrt2},
\quad
\ket{c}=\frac{\ket{1}-\ket{2}}{\sqrt2}.
\label{eq:abc}
\end{eqnarray}
In this basis, we can write $\hH$ as a matrix form, such that
\begin{eqnarray}
\frac{\hH}{\hbar\Omega} =
\begin{pmatrix}
0 & \sqrt{2} & 0 \\
\sqrt{2} & 1 & 0 \\
0 & 0 & -1
\end{pmatrix},
\label{eq:Habc}
\end{eqnarray}
so $\ket{c}$ is a dark eigenstate and the $ab$ sector is invariant.

\begin{theorem}[Full global CFP manifold]
The global CFPs of Eq.~(\ref{eq:Mdef}) are exactly
\begin{eqnarray}
\hat\rho^*(\lambda) = \frac{\lambda}{2}\hat{\Pi}_{ab} + \left( 1-\lambda \right) \ketbra{c}{c} \quad (0 \le \lambda \le1),
\label{eq:cfp}
\end{eqnarray}
where $\hat\Pi_{ab} = \ketbra{a}{a} + \ketbra{b}{b}$. The point $\lambda=2/3$ is $\hI/3$.
\end{theorem}
\begin{proof}[Proof sketch]---Eq.~(\ref{eq:cfp}) is diagonal in the measurement blocks and proportional to $\hat{\Pi}_{ab}$ on the invariant $ab$ sector, while $\ket{c}$ is a dark eigenstate. Hence, $[\hrho^*(\lambda), \hH]=0$ and $\hD(\hrho^*)=\hrho^*$. Conversely, $\hM_0(\hrho)=\hrho$ enforces measurement-block structure. Expanding $\hM_{\tau}(\hrho)=\hrho$ near $\tau=0$ forces equal $a,b$ populations and removes the remaining $b$--$c$ coherence. The full algebra is in the \SM.
\end{proof}


\begin{figure}[t]
\includegraphics[width=0.23\textwidth]{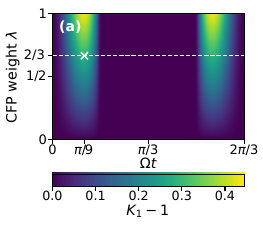}
\includegraphics[width=0.23\textwidth]{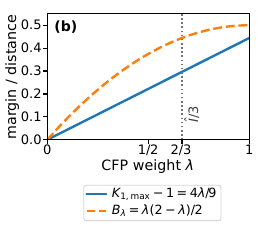}
\caption{Analytic separation on the CFP manifold. (a) Equal-spacing LGI margin $K_1-1$ over time and CFP weight; the white dashed line marks $\lambda=2/3$ and the cross marks $\Omega t=\pi/9$. (b) The exact maximum margin and branch displacement vanish only at the dark state; the vertical line marks $\lambda=2/3$.}
\label{fig:lgi_branch}
\end{figure}

{\em Exact NSIT but analytic LGI violation.}---For every CFP,
\begin{eqnarray}
P^{(j)}(+) = \Tr(\hat{P}_{+} \hrho^*) = \frac{\lambda}{2},
\quad
\eta_{\NSIT}=0.
\label{eq:eta}
\end{eqnarray}
If NSIT certified the temporal classicality, the story would end here. However, it does not. For the equal-spacing protocol $(t_1, t_2, t_3)=(0, t, 2t)$, the $+$ outcome prepares $\ketbra{a}{a}$ and
\begin{eqnarray}
A(t) = \abs{\bra{a}\hU(t)\ket{a}}^2 =\frac{5+4\cos(3\Omega t)}{9}.
\label{eq:A}
\end{eqnarray}
The one-interval probabilities are
\begin{eqnarray}
P_t(++) &=& \frac{\lambda}{2}A(t), \nonumber \\
P_t(+-) &=& P_t(-+) = \frac{\lambda}{2}\left[ 1-A(t) \right], \nonumber \\
P_t(--) &=& 1 - \lambda + \frac{\lambda}{2}A(t).
\label{eq:pair}
\end{eqnarray}
Thus, we have
\begin{eqnarray}
C_{\lambda}(t) = 1 - \frac{8\lambda}{9}\left[ 1 - \cos(3\Omega t) \right],
\label{eq:C}
\end{eqnarray}
and $C_{12} = C_{23} = C_{\lambda}(t)$, $C_{13} = C_{\lambda}(2t)$ give
\begin{eqnarray}
K_1(t) = 1 + \frac{16\lambda}{9}\cos(3\Omega t)\left[ 1 - \cos(3\Omega t) \right],
\label{eq:K}
\end{eqnarray}
with maximum
\begin{eqnarray}
K_{1,\max} = 1+\frac{4\lambda}{9}.
\label{eq:Kmax}
\end{eqnarray}
Figure~\ref{fig:lgi_branch} visualizes this result: the dashed $\lambda=2/3$ slice and the cross at $\Omega t=\pi/9$ lie in the positive-margin region, while panel (b) shows that both $K_{1,\max}-1$ and $B_\lambda$ vanish only at $\lambda=0$.
Every $\lambda>0$ therefore violates an LGI despite exact NSIT; the maximally mixed state gives $K_{1,\max}=35/27$.

\begin{proposition}[Continuum separation]
For every $\lambda \in (0,1]$, $\hrho^*(\lambda)$ satisfies the exact pairwise NSIT and violates an LGI. The dark state $\lambda=0$ is the unique CFP member that is both NSIT-silent and LGI-classical for the above protocol.
\end{proposition}

\begin{proof}---Exact NSIT follows from the CFP theorem. For $\lambda>0$, Eq.~(\ref{eq:Kmax}) exceeds the classical bound. For $\lambda=0$, $Q=-1$ with unit probability in every context, so all LGI facets are obeyed.
\end{proof}

The violation is not disguised signaling. The selective branches are
\begin{eqnarray}
\hrho_{+ | \lambda} &=& \ketbra{a}{a}, \nonumber \\
\hrho_{- | \lambda} &=& \frac{\lambda\ketbra{b}{b} + 2(1-\lambda)\ketbra{c}{c}}{2-\lambda}, \label{eq:branches}
\end{eqnarray}
and $p_+\hrho_{+ | \lambda} + p_-\hrho_{- | \lambda} = \hrho^*(\lambda)$. The nonselective ensemble is exactly restored, but
\begin{eqnarray}
B_{\lambda} &=& \sum_{q=\pm}p_q\frac{1}{2}\norm{\hrho_{q | \lambda} - \hrho^*(\lambda)}_1 =\frac{\lambda(2-\lambda)}{2}.
\label{eq:Blambda}
\end{eqnarray}
Hence, $B_{\lambda}>0$ precisely for the nontrivial CFPs. The maximally mixed point is especially revealing: $\hI/3$ is featureless as an ensemble, but the degenerate measurement filters the $+$ branch into the pure state $\ketbra{a}{a}$, which evolves in the active $ab$ sector and generates temporal correlations incompatible with one classical history.


\begin{figure}[t]
\includegraphics[width=0.23\textwidth]{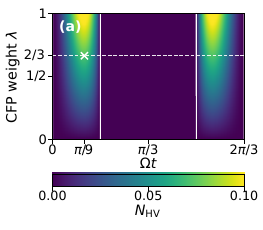}
\includegraphics[width=0.23\textwidth]{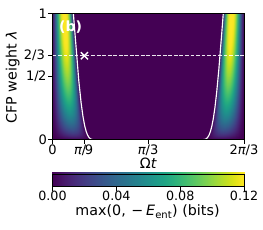}
\caption{Hard hidden-variable and entropic scans. (a) Full-manifold signed hidden-variable negativity $N_{\HV}$ from a $1201\times801$ compatibility calculation, with LP spot checks. (b) Entropic LGI violation strength $\max(0,-E_{\rm ent})$ on the same grid. White contours are zero-obstruction boundaries; dashed lines mark $\hat I/3$, and crosses mark $\Omega t=\pi/9$.}
\label{fig:hv_ent}
\end{figure}

{\em Hidden-variable reconstruction and entropic witness.}---To avoid relying on one correlator facet, we reconstruct the hidden-variable obstruction directly. Let $\Lambda=\{\pm1\}^3$. A context-independent model is a single distribution $P(q_1, q_2, q_3)$ whose pairwise marginals reproduce all measured contexts~\cite{Fine1982,AbramskyBrandenburger2011}. We performed a full-manifold compatibility scan by minimizing the signed hidden-variable negativity
\begin{eqnarray}
N_{\HV} = \min_{\mu : A\mu = P_{\rm obs}}\left[\frac{\norm{\mu}_1-1}{2}\right],        \label{eq:nhv}
\end{eqnarray}
where $\mu$ is a quasiprobability over the eight deterministic assignments and $A\mu=P_{\rm obs}$ imposes all pairwise probabilities. For this consistent three-binary-variable marginal problem the linear-program optimum equals $\max[0, \max_mK_m-1]/4$; we verified this equality by direct LP checks on representative grid points. Fig.~\ref{fig:hv_ent}(a) shows a $1201\times801$ compatibility scan: the obstruction is zero exactly in the classical regions and finite throughout the LGI-violating sectors. The white contours are $N_{\HV}=0$ boundaries, and the dashed $\lambda=2/3$ line with the cross at $\Omega t=\pi/9$ highlights the maximally mixed CFP; LP details are given in the \SM.

As a complementary, non-facet diagnostic, let $P^c_{\HV}$ be an eight-outcome completion of context $c$. The feasible set $\C(Q)$ fixes each observed pairwise marginal and the omitted one-time marginal. We minimize the Jensen--Shannon spread
\begin{eqnarray}
\Delta_{\JS} = \min_{\{P^c_{\HV}\} \in \C(Q)}\frac{1}{3}\sum_c D_{\KL}\left(P^c_{\HV} \| \bar{P}\right),
\label{eq:js}
\end{eqnarray}
where $\bar{P} = \frac{1}{3}\sum_c P^c_{\HV}$. Thus, $\Delta_{\JS}=0$ means one common completion exists, while a positive value measures residual context-dependence~\cite{Lin1991}; details are in the \SM. On representative slices, $\Delta_{\JS}$ peaks at $5.66 \times 10^{-3}$ for $\lambda=1/2$ and $7.75 \times 10^{-3}$ for $\lambda=2/3$.

The conclusion is not tied to correlators. The entropic LGI margin
\begin{eqnarray}
E_{\rm ent} = H(Q_3 | Q_2) + H(Q_2 | Q_1) - H(Q_3 | Q_1)
\label{eq:ent}
\end{eqnarray}
must be nonnegative in a classical temporal model~\cite{UshaDevi2013, Bang2018}. Fig.~\ref{fig:hv_ent}(b) plots $\max(0, -E_{\rm ent})$ over the full CFP-time grid; the white contour is the entropic boundary, and the dashed $\lambda=2/3$ slice enters finite, although narrower, nonclassical windows.


\begin{figure}[t]
\includegraphics[width=0.23\textwidth]{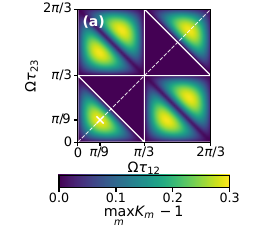}
\includegraphics[width=0.23\textwidth]{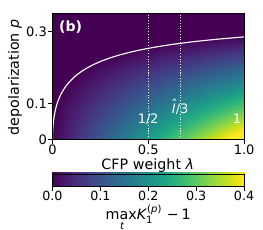}
\caption{Numerical stress tests. (a) Full two-interval landscape for the maximally mixed CFP; the dashed line is the equal-spacing slice and the cross is $\Omega\tau_{12}=\Omega\tau_{23}=\pi/9$. (b) Optimized depolarizing-noise phase diagram. The white curve is the classical boundary; dotted cuts indicate the representative CFP weights $\lambda=1/2,2/3,1$.}
\label{fig:protocol_noise}
\end{figure}

{\em Protocol landscape and robustness.--} The equal spacing is not a fine-tuned slice. For independent intervals $\tau_{12}$ and $\tau_{23}$, we scan all four LGI facets,
\begin{eqnarray}
\begin{array}{ll}
K_1=C_{12}+C_{23}-C_{13},	& K_2=C_{12}-C_{23}+C_{13},\\
K_3=-C_{12}+C_{23}+C_{13},	& K_4=-C_{12}-C_{23}-C_{13},
\end{array}
\label{eq:facets}
\end{eqnarray}
with $C_{13}=C_{\lambda}(\tau_{12} + \tau_{23})$. Fig.~\ref{fig:protocol_noise}(a) shows $\max_mK_m-1$ for $\lambda=2/3$ on a $1001 \times 1001$ grid. The white contours are classical boundaries; the dashed diagonal and cross identify the equal-spacing point $\Omega\tau_{12}=\Omega\tau_{23}=\pi/9$, which lies inside a finite violation island.

We next insert isotropic depolarization after each interval,
\begin{eqnarray}
\hE_p(\hrho) = (1-p)\hrho + \frac{p}{3}\hI,
\quad
\hat\Lambda_t = \hE_p \circ \hcalU_t.
\label{eq:noise}
\end{eqnarray}
Because $\hI/3=\hrho^*(2/3)$ is itself a CFP, this is a natural stress test of the branch-level temporal structure. Optimizing $K_1^{(p)}(t)$ over $1801$ time points on a $451 \times 451$ $(\lambda, p)$ grid gives Fig.~\ref{fig:protocol_noise}(b). The white curve is the optimized LGI boundary, and the dotted cuts give the representative thresholds below; noisy-channel and sampling details are in the \SM. The critical depolarization strengths are
\begin{eqnarray}
 p_c(\lambda=1/2) &\simeq& 0.254, \nonumber\\
 p_c(\lambda=2/3) &\simeq& 0.268, \nonumber\\
 p_c(\lambda=1) &\simeq& 0.286.
 \label{eq:pc}
\end{eqnarray}
Finite-shot Monte Carlo sampling at $\lambda=2/3$, $p=0.1$, and $\Omega t=\pi/9$ gives the true noisy value $K_1=1.178$; the $5$--$95\%$ band clears the classical bound at about $N=200$ shots per pairwise context (for more details, see the \SM).


{\em Discussion and conclusion.}---We have shown that exact NSIT can be a fixed-point shadow of a nonselective channel rather than evidence for a classical temporal process. The ensemble can look perfectly innocent precisely, because the invasive branches have been recombined; NSIT is therefore a statement about the channel after outcome erasure, not directly about a macrorealist trajectory. The qutrit ring provides a complete analytic separation: the full CFP manifold is known; pairwise NSIT vanishes identically on it; every nontrivial member violates an LGI with $K_{1,\max}=1+4\lambda/9$; the violation is organized by the branch displacement $B_\lambda=\lambda(2-\lambda)/2$; and hidden-variable, entropic, protocol, noise, and sampling analyses all locate the same branch-level obstruction.

This clarifies what earlier NSIT--LGI results did not settle. Complete NSIT-type packages can characterize the macrorealism, and LGIs alone are not a complete temporal analogue of Bell inequalities~\cite{ClementeKofler2015,ClementeKofler2016,Halliwell2017}. However, those logical statements do not justify interpreting the exact pairwise marginal NSIT as the branch-level noninvasive dynamics. Macrorealist noninvasiveness concerns access to a history, whereas pairwise NSIT concerns one class of marginals; CFPs show how these ideas separate in quantum dynamics. Likewise, prior demonstrations of LGI violations under restricted or ambiguous signaling conditions showed that signaling cannot always explain the temporal nonclassicality~\cite{Emary2017,Wang2018}. Our result goes further: it gives an explicit fixed-point mechanism, a full analytic state manifold including $\hI/3$, an exact zero NSIT residual, and a quantified branch resource responsible for the remaining contextual obstruction. An exactly silent marginal can therefore be the shadow of a disturbance rebalanced at the channel level.

The implication is immediate for temporal quantum information. Discarding an outcome register can erase every marginal trace of the measurement backaction while leaving the branch-conditioned histories incompatible with any context-independent model. NSIT is therefore an operational certificate of no marginal signaling after outcome erasure, not a certificate of branch-level non-disturbance or temporal classicality. Future work can thus use CFP engineering as a diagnostic tool: by deliberately making ensembles silent, one can expose the quantum temporal structure that survives only in conditioned branches.

\begin{acknowledgments}
{\em Acknowledgments.}---{We are grateful to Professor  C. Brukner for insightful comments on an earlier version of this manuscript and Professor  J. Yi for valuable discussions that helped formulate the initial idea of this work.} This work was supported by the Institute of Information and Communications Technology Planning and Evaluation (IITP) (RS-2024-00396999, IITP-2026-RS-2020-II201606, RS-2022-II221029, and RS-2025-25464481) and the National Research Foundation of Korea (RS-2023-00283146, RS-2024-00432214, RS-2024-00432214, RS-2025-03532992, and RS-2025-25454381).
\end{acknowledgments}

\bibliographystyle{apsrev4-2}
\bibliography{references_CFP}

\end{document}